%% file: main.tex
\documentclass{article}

\usepackage[utf8]{inputenc}

\usepackage{xcolor,amsmath,amsfonts,amsthm,mathrsfs,fancyhdr,verbatim,framed,hyperref,enumerate}
\usepackage[margin=1.in]{geometry}
\usepackage[capitalize]{cleveref}
\usepackage{mathpazo}
\usepackage{bm}
\usepackage{amssymb}
\usepackage{cancel}
\usepackage{qcircuit}
\usepackage{setspace}
\usepackage{braket}
\usepackage{tikz}
\usepackage{etoolbox}
\usepackage[ruled]{algorithm2e}
\usepackage{tikz}
\allowdisplaybreaks

\usepackage{graphicx}
\usepackage{tcolorbox}
\usepackage{tikz}
\usepackage{braket} 

\usetikzlibrary{arrows.meta,positioning,calc}

\newcommand{\eps}{\varepsilon}

\newcommand{\tps}{\tilde{\psi}}

\newcommand{\ip}[2]{\langle #1 | #2 \rangle}

\newcommand{\abs}[1]{\left\vert {#1} \right\vert}

\newcommand{\C}{\mathbb{C}}

\let\Re\relax
\DeclareMathOperator{\Re}{\mathrm{Re}}

\renewcommand{\cal}[1]{\mathcal{#1}}
\DeclareMathOperator*{\E}{\mathbb{E}}

\newcommand{\PSPACE}{\mathsf{PSPACE}}

\newcommand{\MIP}{\mathsf{MIP}}

\DeclareMathOperator{\poly}{poly}

\DeclareMathOperator{\var}{Var}

\newcommand{\YES}{\mathrm{YES}}
\newcommand{\NO}{\mathrm{NO}}

\newtheorem{theorem}{Theorem}
\newtheorem{definition}[theorem]{Definition}

\newtheorem{corollary}[theorem]{Corollary}
\newtheorem{lemma}[theorem]{Lemma}

\newtheorem{claim}[theorem]{Claim}
\newtheorem{remark}[theorem]{Remark}
\newtheorem{problem}[theorem]{Problem}

\newif\ifnotes
\notesfalse          

\ifnotes
\newcommand{\agi}[1]{\textcolor{blue}{(\'{A}gi: #1)}}
\newcommand{\anote}[1]{\textcolor{red}{(Anand: #1)}}
\newcommand{\avishay}[1]{\textcolor{teal}{(Avishay: #1)}}

\else
\newcommand{\agi}[1]{}
\newcommand{\anote}[1]{}
\newcommand{\avishay}[1]{}
\fi

\usepackage{mathrsfs}

\newcommand{\acc}{\mathrm{acc}}
\newcommand{\rej}{\mathrm{rej}}

\newcommand{\BQP}{\mathsf{BQP}}

\newcommand{\PCP}{\mathsf{PCP}}
\newcommand{\adPCP}{\mathsf{adPCP}}

\newcommand{\IP}{\mathsf{IP}}

\newcommand{\coNP}{\mathsf{coNP}}

\newcommand{\PPT}{\mathrm{PPT}}

\def\cH{{\cal H}}

\def\cO{{\cal O}}

\newcommand{\Forr}{\textsc{Forrelation}}

\newif\ifnicematrix
\nicematrixtrue

\ifnicematrix
\newcommand{\nicematrix}[1]{#1}
\else
\newcommand{\nicematrix}[1]{}
\fi

\ifnicematrix
\usepackage{nicematrix}
\NiceMatrixOptions{
code-for-first-row = \color{blue} ,
code-for-last-row = \color{blue} ,
code-for-first-col = \color{blue} ,
code-for-last-col = \color{blue}
}
\fi

\title{A Relativizing $\MIP$ for $\BQP$}

\newif\ifanon\anonfalse
\ifanon
\author{Anonymous}
\date{}
\else
\author{Scott Aaronson\thanks{\texttt{aaronson@cs.utexas.edu}}\\\small{UT Austin} \and Anand Natarajan\thanks{\texttt{anandn@mit.edu}}\\\small{MIT} \and Avishay Tal\thanks{\texttt{atal@berkeley.edu}}\\\small{UC Berkeley} \and \'Agi Vill\'anyi\thanks{\texttt{agivilla@mit.edu}}\\\small{MIT}}
\date{October 5, 2026}
\fi

\begin{document}
\maketitle

\begin{abstract}

Complexity class containments involving interactive proof classes are famously nonrelativizing: although $\IP = \PSPACE$, Fortnow and Sipser showed that that there exists an oracle relative to which $\coNP \not\subseteq \IP$. In contrast, the question of whether the containment $\BQP \subseteq \IP$ is relativizing remains wide open. In this work we make progress towards resolving this question by showing that the containment $\BQP \subseteq \MIP$ holds with respect to any classical oracle. We obtain this result by constructing, for any classical oracle $\cO$, a $\PCP$ proof system for $\BQP^{\cO}$ where the verifier makes polynomially many classical queries to an exponentially-long proof, and to the oracle $\cO$. Our construction is inspired by the state synthesis algorithm of Grover and Rudolph, and serves as a complement to the ``exponential PCP'' constructed by Aharonov, Arad, and Vidick, which achieves similar parameters but which is based on different ideas and does not relativize. We propose relativization as a proxy for prover efficiency, and hope that progress towards an $\IP$ for $\BQP$ in the oracle world will lead to a non-cryptographic interactive protocol for proving any quantum computation to a classical skeptic in the unrelativized world, which is a longstanding open problem in quantum complexity theory.

\end{abstract}
\ifanon


\fi

\section{Introduction}

A major open question in quantum complexity is whether every language in $\BQP$ has an interactive proof system with a polynomial-time classical verifier and a polynomial-time quantum prover (a ``doubly efficient'' interactive proof system); in symbols, whether $\BQP = \IP_{\BQP}$. This question comes from a natural motivation: can a quantum computer certify its correctness to a classical verifier? 
Given an interactive proof system, consider the exponentially large table of all prover responses to all possible challenges by the verifier. With access to such a table, the verifier is able to be convinced by ``spot checking'' the table in only a polynomial number of locations. So such a proof for $\BQP$ would imply that quantum computations have a classical representation (of possibly exponential size) that can be \emph{locally checked}.

Setting aside the cryptographic argument systems of Mahadev~\cite{mahadev2018classical} and others (in which the soundness of the proof depends on a cryptographic assumption), the best positive partial progress we have towards this problem is the following two results. First, the standard results $\BQP \subseteq \PSPACE$ and $\PSPACE = \IP$ imply that $\BQP \subseteq \IP$. However, the protocol obtained this way is undesirable for several reasons. Firstly, it is not ``doubly efficient'': the honest prover must actually have power far exceeding that of a quantum polynomial-time machine to in order to successfully execute the honest strategy. This is closely related to the second problem: the protocol essentially discards all the ``quantum structure'' of a $\BQP$ computation, by reducing it to an exponentially large sum of amplitudes, and then applies powerful algebraic techniques (namely, the sum-check protocol) to obtain an $\IP$ protocol. 

Second, in the $\MIP^*$ model of \emph{multi-prover} interactive proofs with quantum entangled provers, doubly efficient protocols for $\BQP$ are known, thanks to the early works of McKague~\cite{mckague-2011-interactiveproofsefficientquantum}, Reichardt, Unger, and Vazirani~\cite{reichardt2013classical}, and the later streamlined one-round protocol of Grilo~\cite{grilo2019simple}. In these protocols, each prover only needs to implement a polynomial-time quantum circuit and act on polynomially many shared entangled pairs. However, although the messages sent between the provers and the verifier are classical, it is hard to identify any classical representation of the computation that is being locally checked. Indeed, the protocols rely on the quantum phenomenon of \emph{self-testing}, in which classical correlations between the provers certify that they physically perform certain quantum operations during the protocol. 

In this work, we would like to propose a new axis for measuring progress towards $\BQP = \IP_{\BQP}$: that of \emph{relativization}. Classically, as we have alluded to above, several different types of local checkability are known: the simple Cook-Levin construction for 3SAT, which is doubly efficient, as well as more sophisticated algebraic techniques such as sum-check, which are responsible for $\IP = \PSPACE$ and the $\PCP$ theorem. Fortnow~\cite{fortnow1994role} argued that the latter are distinguished from the former by failing to relativize\footnote{Whether the Cook-Levin construction itself relativizes was disputed by~\cite{AIV92}, but \cite{fortnow1994role} shows that it does under an appropriate oracle access model.}. It is known that the complexity class containment $\PSPACE \subseteq \IP$ fails to hold relative to an oracle: indeed, Fortnow and Sipser showed~\cite{fortnow-1988-AreTI} that there exists an oracle such that $\coNP^{O} \not\subseteq \IP^O$. In the quantum world, while no oracle separating $\BQP$ from $\IP$ is known, it appears that all known $\IP$ and $\MIP^*$ protocols for $\BQP$ fail to relativize. 
    \begin{itemize}
        \item The basic sum-check protocol for $\BQP \subseteq \IP$ fails to relativize for the standard reasons: in the presence of an oracle, it is not possible for the verifier to efficiently compute the finite-field extension of the function implemented by the oracle.
\item An alternative protocol for $\BQP \subseteq \IP$ was proposed by Aharonov, Arad, and Vidick~\cite{aharonov2013guest}. This protocol does less violence to the quantum structure of the $\BQP$ computation, but still fails to relativize. In this protocol, the verifier is asked to convert the quantum circuit to a Hamiltonian $H$ using the Feynman-Kitaev construction, and then to evaluate diagonal entries of the Hamiltonian in a product state basis, i.e., quantities of the form
\[ \bra{\psi_1}\bra{\psi_2} \dots \bra{\psi_n} H \ket{\psi_1} \ket{\psi_2} \dots \ket{\psi_n},\]
for arbitrary single-qubit states $\ket{\psi_1}, \dots, \ket{\psi_n}$. For $H$ obtained from a circuit with oracle gates, these entries cannot be computed by a verifier that only makes polynomially many queries to the oracle: e.g., if one takes each $\ket{\psi_i} = \ket{+}$, then the quantity to be computed depends on the sum of all the entries of the oracle. 
\item The $\MIP^*$ protocols of \cite{mckague-2011-interactiveproofsefficientquantum,reichardt2013classical,grilo2019simple} fail to relativize because they rely on self-tests for local quantum operations, and these are not known to generalize to oracle gates.
    \end{itemize}

The main result of our work is the existence of an $\MIP$ protocol for $\BQP$ that \emph{does} relativize. Note that $\MIP$ is a classical complexity class---the provers are allowed no shared entanglement, unlike $\MIP^*$---and our result can alternately be stated as showing the existence of an \emph{exponentially long classical $\PCP$} for $\BQP$, where the $\PCP$ proof string can be checked with a polynomial number of queries. Somewhat more formally, our results are as follows.
\begin{theorem}[Theorem~\ref{thm:pcp-for-bqp} in the body]\label{thm:pcp-for-bqp_intro}
For any oracle $\cO$, there is an adaptive $\PCP$ system for any promise problem in $\BQP^{\cO}$ with perfect completeness and inverse-exponential soundness, where the proof is an exponentially long classical string, and the verifier is a classical PPT machine with oracle access to $\cO$ that can flip $\poly(n)$ random coins and query $\poly(n)$ locations (adaptively) in the $\PCP$ proof.
    \end{theorem}

    We remark that the usual $\PCP$ systems have non-adaptive queries to the proof (e.g., \cite{arora1998probabilistic, Arora-1998-PCP}). Indeed, for query complexity $q = O(1)$ over constant-size alphabet $\Sigma$, one can convert an adaptive $\PCP$ to a non-adaptive $\PCP$ by ``querying all possible paths along the decision tree'',  yielding $|\Sigma|^q = O(1)$ non-adaptive queries. This conversion, though, is meaningless when $q = \poly(n)$ and the proof length is exponential, as it yields a non-adaptive $\PCP$ where the verifier may read the entire proof. For that reason, we do not perform the conversion, but instead convert adaptive $\PCP$ systems directly to $\MIP$ protocols. 
    
    \begin{theorem}[Theorem~\ref{thm:MIP_for_BQP} in the body]
        \label{thm:MIP_for_BQP_intro}
    For any oracle $\cO$:
    \begin{equation}
        \BQP^\cO \subseteq \MIP^\cO
    \end{equation}
    \end{theorem}

The basic idea of our $\PCP$ is to have the $\PCP$ proof contain representations of the quantum state $\ket{\psi_t}$ of the computation at each timestep $t$, where we apply only a single local quantum gate in each timestep. However, crucially, we do not simply write out the components of the state vector $\ket{\psi_t}$: instead, we write down the truth table of a classical state synthesis oracle for $\ket{\psi_t}$. This state synthesis oracle has the property that a \emph{classical} algorithm can use it to \emph{sample} from the distribution obtained by measuring the state in the computational basis, using only a polynomial number of queries. We show that this sampling access enables us to estimate inner products between states, and thus to certify that the provided proof corresponds to a correct computational history.
The entries of the state synthesis oracle are the marginal probabilities of the distribution induced by the states, as well as the phases of the leaves.

We also remark on another striking feature of our result: it illustrates a dimension in which $\BQP$ is much more ``tame'' than the polynomial hierarchy. It is conventionally believed that $\BQP$ is not contained in the polynomial hierarchy, and Raz and Tal~\cite{raz-tal-ph} showed an oracle relative to which this is true. But our result shows that $\BQP$ is always contained in $\MIP$, whereas Fortnow and Sipser's oracle \cite{fortnow-1988-AreTI} can easily be seen to separate $\coNP$ from not just $\IP$, but $\MIP$ as well. 

\section{Discussion and Future Work}

Our result has raised several interesting questions for future work. Most directly, it drew attention to the open question of finding an oracle separating $\BQP$ from $\IP$, or proving a relativizing containment \cite{aaronson2021open}; after the appearance of our work, such a separation was found by~\cite{bouland2026}. That result, combined with our work, gives the first known oracle separation between $\IP$ and $\MIP$. 

Two natural follow-up directions are to reduce the computational power of the prover in our protocol, and to investigate whether efficient-prover protocols in other proof models relativize. We know from \cite{reichardt2013classical, grilo2019simple} that there exists an efficient-prover $\MIP^*$ for $\BQP$. Can this protocol be made to relativize, and could the ideas let us find an $\MIP$ for $\BQP$, where the provers require only the power of $\BQP$ to establish completeness? Similarly, is there a relativizing cryptographic argument system for $\BQP$ with an efficient (single) prover?

At a more qualitative level, an interesting direction is to revisit the question of whether the local checkability of computation relativizes. In the classical setting, this question was the subject of the ``Great Oracle Debate of 1993''~\cite{fortnow2009oracle}, carried out in written form in the works of Arora, Impagliazzo and Vazirani~\cite{AIV92} and Fortnow~\cite{fortnow1994role}. These works highlighted that whether local checkability relativizes depends on the nature of the checker. Specifically, \cite{AIV92} showed that the local checkability of polynomial-time computations by a log-space verifier, which is guaranteed by the Cook-Levin theorem, fails to hold relative to oracles---simply because a log-space verifier cannot access the oracle on inputs of polynomial length. In the quantum setting, one could equally well ask whether the local checkability of quantum computation guaranteed by the Feynman-Kitaev construction relativizes. In our view, the crux of the question is once again the nature of the power granted to the checker. Specifically, we interpret the failure to relativize of the protocols of~\cite{aharonov2013guest} and~\cite{grilo2019simple}, which both use the Feynman-Kitaev construction to give a classical-verifier protocol for $\BQP$, as evidence that the Feynman-Kitaev construction is non-relativizing with respect to classical verifiers. In contrast, our $\MIP$ protocol shows that there is some sense in which classical local checkability of quantum computation is relativizing: this raises the question of whether there exist alternate circuit-to-Hamiltonian constructions, perhaps along the lines of~\cite{ABN23}, that do relativize.

\section{Technical Overview}

In this section, we explain our proof techniques by writing out the exponential $\PCP$ protocol explicitly for the $2$-fold $\Forr$ problem from \cite{aa2018forrelation}. This problem was shown in~\cite{aa2018forrelation} to require exponentially many classical queries to solve, and is a special case of the $k$-fold $\Forr$ which is $\BQP$-complete, making it a good simplified setting to explain the main ideas of our protocol.

Recall that in $2$-fold $\Forr$, we are given oracle access to boolean functions $f_1, f_2 : \{0,1\}^n  \rightarrow \{0,1\}$. The \emph{forrelator} of $f_1, f_2$ is the quantity:  

\begin{equation}
   \Phi_{f_1, f_2} = \frac{1}{2^{3n/2}} \sum_{x, y \in \{0,1\}^n} (-1)^{f_1(x)} \cdot (-1)^{x \cdot y} \cdot (-1)^{f_2(y)}
\end{equation}

The goal of $2$-fold $\Forr$ is to decide whether $\Phi_{f_1, f_2} \geq 0.6$ or $|\Phi_{f_1, f_2}| \leq 0.01$, given the promise that one is the case. This problem can be solved quantumly by the simple circuit in Figure~\ref{fig:two-fold-forrelation}. We will now describe our $\PCP$  system for this problem, which is based on this quantum circuit. 






\input{figures/forrelation-circuit}
\input{figures/forrelation-circuit-sequential}
This circuit (described in more detail in Appendix~\ref{app:forr}), consists solely of Hadamard gates and oracle gates for the functions $f_1, f_2$, which we denote by $\cO_{f_1}, \cO_{f_2}$ respectively, with a total of $m = \poly(n)$ gates applied. The circuit \emph{accepts} if the final measured output is equal to $0^n$, and \emph{rejects} otherwise\footnote{Note that this is slightly different from our convention for general circuits, which will be to have acceptance encoded by a single bit of the output.}.

\paragraph{The structure of the proof.}

In our $\PCP$ system, the prover will be expected to send an exponentially long proof string $\pi$. This proof will consist of segments indexed by $i \in \{0, 1, \dots, m\}$, each of which consists of a classical description (in a sense to be described later) of a quantum state $\ket{\tps_i}$. An honest prover will choose $\ket{\tps_i}$ to be equal to the state of the quantum circuit in Figure~\ref{fig:two-fold-forrelation-sequential} immediately after the $i$th gate has been applied (with $\ket{\tps_0}$ being the initial state of the computation). The verifier will execute tests on the proof that query $\poly(n)$ bits and check that the sequence of states in the proof corresponds to a valid computational history ending in an accepting configuration (in which the accepting output $0^n$ has probability at least $2/3$).

The na\"{i}ve representation of quantum states as a list of coefficients in the state vector does not suffice for these tests. Instead, we use a different representation, based on the state-synthesis algorithm of~\cite{aaronson2016complexityquantumstatestransformations, grover2002creatingsuperpositionscorrespondefficiently}. For a state
\[ \ket{\tps} = \sum_{x \in \{0,1\}^n} \tilde{\alpha}_x \ket{x},\]
observe that we can write each coefficient $\tilde{\alpha}_x$ as a product

\[ \tilde{\alpha}_x = \gamma_x \cdot \sqrt{ p_{x_1}\cdot  p_{x_2|x_1} \cdot p_{x_3|x_1 x_2}  \cdots p_{x_n |x_1 \dots x_{n-1}}},\]
where $\gamma_x$ is a complex phase (a unit-modulus complex number), and the quantities $p_{x_k | x_{1} \dots x_{k-1}}$ are conditional probabilities. The representation of $\ket{\tps}$ that we will use in our proof is as a table $(\gamma, p)$ of complex phases and conditional probabilities, represented as binary numbers up to $\poly(n)$ bits of precision. This can be visualized as a prefix-tree, as depicted in Figure~\ref{fig:aaronson-prefix-tree2}.  This representation has the following nice properties:

\input{figures/prefix-tree2}

\begin{itemize}
    \item \textbf{Validity:} By suitably omitting redundant probabilities, \emph{every} proof string can be interpreted as a valid, normalized quantum state.
    \item \textbf{Sampling access:} Using $O(n)$ adaptive queries to the proof, the verifier can sample a string $x$ with probability $|\tilde{\alpha}_x|^2$.
    \item \textbf{Amplitude access:} Using $O(n)$ queries to the proof, for any given $x$, the verifier can compute $\tilde{\alpha}_x$.
\end{itemize}

\paragraph{The verifier's tests.}
Given a proof $\pi$ specifying a state sequence $\{\ket{\tps_i}\}_{0 \leq i \leq m}$, the verifier will execute tests that check that (1) the claimed initial state $\ket{\tps_0}$ matches the desired state $\ket{0^n}$, (2) each pair of states $\ket{\tps_{i-1}}$ and $\ket{\tps_{i}}$ are correctly related by the application of the $i$th gate $G_i$, and (3) the final state gives a high probability to the desired outcome $0^n$. Tests (1) and (3) are fairly straightforward, so in this overview, we will focus on test (2).

The main idea behind our test is the following fact: let $\ket{\psi} = \sum_x \psi_x \ket{x} $ and $\ket{\phi} = \sum_x \phi_x \ket{x}$ be states, and let $X \sim \psi$ denote a random string $X \in \{0,1\}^n$ sampled by measuring $\ket{\psi}$ in the computational basis. Then if $\ket{\psi}$ and $\ket{\phi}$ are $\delta$-far apart (in that $|1 - \braket{\psi| \phi}| > \delta$), then sampling $X \sim \psi$, it holds with probability $\Omega(\delta^2)$ that the coefficients $\psi_X$ and $\phi_X$ of the two states at $X$ are detectably different at $O(n)$ bits of precision (i.e. $|\phi_x - \psi_X| \geq 2^{-O(n)}$). Thus, if we are given $\ket{\psi}$ and $\ket{\phi}$ in our description format, we can check that they are close by the following procedure.
\begin{itemize}
    \item Repeatedly sample $X \sim \psi$ using the sampling access to $\ket{\psi}$.
    \item For each sampled value $X$, compute $\psi_X$ and $\phi_X$ using amplitude access. Reject if they are not equal.
\end{itemize}

This basic test certifies that two states are nearly equal. But what we want in our protocol is to check that the states of two adjacent timesteps are related by the appropriate gate:
\[ \ket{\tps_i} \approx G_i \ket{\tps_{i-1}}.\]

We obtain this using a simple variant of the protocol above. Observe that each gate $G_i$ is either a local gate, acting on at most $2$ qubits of the state, or an oracle gate. If $G_i$ is a local gate, every coefficient of the state $G_i \ket{\tps_{i-1}}$ can be computed by reading at most 4 coefficients of the state $\ket{\tps_{i-1}}$. On the other hand, if $G_i$ is an oracle gate, then because oracle gates are diagonal, each coefficient in $G_i \ket{\tps_{i-1}}$ can be computed from just \emph{one} coefficient in $\ket{\tps_{i-1}}$ together with one query to the oracle:
\[ \bra{x} \cO_f \ket{\tps_{i-1}} = (-1)^{f(x)} \braket{x | \tps_{i-1}}. \]

Thus, our test is the following:
\begin{itemize}
    \item Repeatedly sample $X \sim \tps_{i}$ using sampling access to $\ket{\tps_{i}}$.
    \item For each sampled value $X$, compute the $X$-component of $\ket{\tps_i}$ using amplitude access. Compute the $X$-component of $G_i \ket{\tps_{i-1}}$ by performing up to four amplitude access queries to $\ket{\tps_{i-1}}$. Reject if these coefficients differ.
\end{itemize}

For a suitable polynomial number of samples, one can show that this test rejects with probability exponentially close to $1$ if $\ket{\tps_i}$ and $G_i\ket{\tps_{i-1}}$ are $1/\poly(n)$ far from each other.

\paragraph{Putting together the $\PCP$.}

Putting the pieces together, the verifier of our $\PCP$ proof system performs the following operations. Given a proof $\pi$ consisting of a description of a sequence of states $\ket{\tps_0}, \ket{\tps_1}, \dots, \ket{\tps_m}$, the verifier first checks that $\ket{\tps_0}$ is close to the desired initial state $\ket{0^n}$, using polynomially many queries to $\pi$. Next, for each $i \in \{1, \dots, m\}$, the verifier checks that $\ket{\tps_{i}} \approx G_i \ket{\tps_{i-1}}$, where $G_i$ is the $i$th gate in the circuit, using the test described above: this requires $\poly(n)$ adaptive accesses to the proof $\pi$, as well as possibly $\poly(n)$ queries to $\cO$. Finally, the verifier checks that the final state $\ket{\tps_{m}}$ assigns a sufficiently high probability to the desired outcome---in this case, $0^n$. This again requires $\poly(n)$ queries to $\pi$. This yields a complete and sound proof system for 2-fold $\Forr$. \anote{What we ended up writing was completely generic}


\input{figures/mip-colored}

\paragraph{Obtaining an $\MIP$ with an ``adaptive'' clause-variable transformation.}
To convert this $\PCP$ to an $\MIP$ proof system, we use a variant of the standard clause-variable transformation. Typically, this transformation works as follows: given a system of local constraints to be checked on some proof string $\pi$, we define a two-prover, one-round $\MIP$ protocol where the verifier samples a a single constraint (a ``clause'') uniformly at random from the system, and sends the label of this constraint to the first prover, and the label of a single  location in $\pi$ touched by the constraint (a ``variable'') to the second prover. The verifier expects the first prover to respond with an assignment to all the locations in the proof touched by the constraint, and the second prover to respond with an assignment to the single requested location. It accepts if the first prover's assignment satisfies the constraint, and is consistent with the second prover's assignment.

In our case, there is a subtlety caused by the \emph{adaptivity} of the $\PCP$ verifier. This means that, even for a fixed setting of the $\PCP$ verifier's random coins, we cannot predict which locations the verifier will read in $\pi$ without knowing $\pi$ itself, since each location read by the verifier is a function of its random coins as well as the values in the previously read locations of $\pi$. To handle this, we introduce a two-round variant of the clause-variable transformation, depicted in Figure~\ref{fig:mip-colored}. In this version, the $\MIP$ verifier first samples a setting $R$ of the $\PCP$ verifier's random coins, and sends this to the first prover $P_1$. $P_1$ responds with an assignment $(a_1, \dots, a_{q_\pi})$, ostensibly corresponding to the values in the locations of $\pi$ that the $\PCP$ verifier would have queried on randomness $R$. Now, by simulating the $\PCP$ verifier on $R$ and the values $a_1, \dots, a_{q_\pi}$, the $\MIP$ verifier can determine which locations $i_1, \dots, i_{q_\pi}$ the $\PCP$ verifier would have queried. In the next round of the $\MIP$, the verifier picks a random location $i_j$ out of $i_1, \dots, i_{q_\pi}$ and asks the second prover $P_2$ for an assignment to this location. The $\MIP$ verifier accepts if $P_2$'s assignment is consistent with $P_1$'s assignment on the location $i_j$, and $P_1$'s assignment leads the simulated $\PCP$ to accept on $R$.

The analysis of this version of the clause-variable transformation is quite similar to that of the standard version: one extracts a proof string $\pi$ from $P_2$'s responses to all possible locations $i$, and argues that $\pi$ must have been accepted by the $\PCP$ verifier if $P_1, P_2$ were accepted by the $\MIP$ verifier. 

\section{Preliminaries}
\subsection{Quantum states}

\begin{definition}[State sampling] For any quantum state $\ket{\psi}$:

\begin{equation}
    \ket{\psi} = \sum_{x \in \{0,1\}^n} \alpha_x \ket{x},
\end{equation}

we denote by $X \sim \psi$ the distribution formed by sampling $X$ with probability $|\alpha_X|^2$. 
    
\end{definition}

\subsection{Proof Systems}
\begin{definition}[Quantum Queries to Boolean functions]

For any boolean function $f : \{0,1\}^* \rightarrow \{0,1\}$, let $\cO_f$ be the corresponding classical oracle. Quantum queries to $\cO_f$ perform the following unitary transformation: 

\begin{equation*}
    \cO_f \ket{x} = (-1)^{f(x)} \ket{x}
\end{equation*}
    
\end{definition}

We note that any quantum algorithm that makes controlled queries to a classical oracle $\cO_f$ can be simulated by one that makes uncontrolled queries to a classical oracle $\cO_g$ for the Boolean function $g$ that is the classically controlled version $f$. So in our definition of $\BQP^{\cO}$, we will only allow the quantum algorithm to make uncontrolled queries to $\cO$.

\begin{definition}[Oracle $\BQP$] Let $\cO$ be any classical oracle. We say a quantum circuit $C_n$ \emph{accepts} an input $x$ if measuring the first qubit of the output state yields $0$. A promise problem $L = (L_\YES, L_\NO)$ is in $\BQP^{{\cO}}$ if and only if there exists a uniform family of quantum circuits $\{C_n\}$ such that, for every $x \in \{0,1\}^n$:

\begin{itemize}
    \item if $x \in L_\YES$, then $\Pr[\acc \leftarrow C_n(x)] \geq 2/3$.
    \item if $x \in L_\NO$, then $\Pr[\acc \leftarrow C_n(x)] \leq 1/3$
\end{itemize}
    
\end{definition}


\begin{definition}[Oracle $\MIP$] Let $\cO$ be any classical oracle. A language $L^\cO$ has an \emph{oracle multi-prover interactive proof} ($\MIP^\cO$) if there exists an interaction between a $\PPT$ algorithm $V^\cO$ and $k$ (potentially computationally unbounded) non-communicating provers $(P^\cO_1, \dots, P^\cO_k)$ such that, for all $x \in \{0,1\}^n$:

\begin{enumerate}
    \item Efficiency: $V$ takes as input $x \in \{0,1\}^n$ and outputs $\acc$ or $\rej$, using at most $\poly(n)$ random coins and making at most $\poly(n)$ combined queries to the $k$ provers and the oracle $\cO$. We denote by $V^\cO(x)$ the random variable corresponding to the output of $V$ on input $x$, with oracle access to $\cO$.  
    \item Completeness: If $x \in L^\cO$: 

    \begin{equation*}
        \Pr\left[\acc \leftarrow \left[V^\cO(x), (P_1^\cO, \dots, P_k^\cO)\right]\right] \geq 2/3.
    \end{equation*}

    \item Soundness: If $x \notin L^\cO$, then for all non-communicating provers $(\hat{P}_1, \dots \hat{P}_k)$: 

    \begin{equation*}
        \Pr\left[\acc \leftarrow \left[V^\cO(x), (\hat{P}_1^\cO, \dots, \hat{P}_k^\cO)\right]\right] \leq 1/3.
    \end{equation*}
\end{enumerate}
    
\end{definition}

\begin{definition}[Oracle Adaptive $\PCP$s] Let $\cO$ be any classical oracle. An oracle language $L^\cO$ has an \emph{oracle adaptive probabilistically checkable proof} with soundness $s$ and completeness $c$ ($\adPCP^\cO_{s, c}[p(n), q(n), r(n)]$) if there exists a $\PPT$ algorithm $V$ satisfying the following properties for all $x \in \{0,1\}^n$:

\begin{enumerate}
    \item \emph{Efficiency:} Given random access to a string $\pi \in \{0,1\}^*$ and an oracle $\cO$, $V$ takes as input $x \in \{0,1\}^n$ and outputs $\acc$ or $\rej$, using at most $r(n)$ random coins, making at most $p(n)$ adaptive queries to locations of $\pi$, and at most $q(n)$ adaptive queries to $\cO$. We denote by $V^{\pi, \cO}(x)$ the random variable corresponding to the output of $V$ on input $x$ with random access to $\pi$ and oracle access to $\cO$.
    \item \emph{Completeness:} If $x \in L^\cO$, there exists a proof $\pi \in \{0,1\}^*$ such that:
    \begin{equation*}
        \Pr[V^{\pi, \cO}(x) = \acc] \geq c.
    \end{equation*}
    \item \emph{Soundness:} If $x \notin L^\cO$, then for all $\pi \in \{0,1\}^*$:
    \begin{equation*}
        \Pr[V^{\pi, \cO}(x) = \acc] \leq s.
    \end{equation*}
\end{enumerate}
\end{definition}

\subsection{Forrelation}

\begin{definition}[Forrelator]
    Fix any two boolean functions $f_1, f_2 : \{0,1\} \rightarrow \{0, 1\}$.  The \emph{forrelator} of $f_1$ and $f_2$ is the quantity:

\begin{equation*}
  \Phi_{f_1, f_2}
  = \frac{1}{2^{3n/2}}
     \sum_{x,y\in\{0,1\}^n} (-1)^{f_1(x)}(-1)^{x\cdot y}(-1)^{f_2(y)}.
  \label{eq:forr-phi}
\end{equation*}
\end{definition}

We consider the usual promise version of the Forrelation problem from \cite{aa2018forrelation}.

\begin{problem}[Forrelation]\label{pr:forrelation}
Fix any two boolean functions $f_1, f_2 : \{0,1\}^* \rightarrow \{0, 1\}$. Given query access to $f_1$ and $f_2$ and the promise that the forrelator is
either $\Phi_{f_1, f_2} \ge 0.6$, or $\abs{\Phi_{f_1, f_2}} \le 0.01$,
decide which is the case.
\end{problem}

\section{A Relativizing Exponential $\PCP$ for $\BQP$}
In this section, we prove our main theorem. 

\begin{theorem}\label{thm:pcp-for-bqp}
For any oracle $\cO$:
\[\BQP^\cO \subseteq \adPCP^{\cO}_{s, c}[\poly(n), \poly(n),\poly(n)],\]

where $s = 2^{-\poly(n)}$ and $c = 1$.
\end{theorem}
In fact, the number of rounds of adaptivity of the $\adPCP$ verifier can be reduced to $2$ at the cost of reducing the completeness-soundness gap to inverse-polynomial, as will be shown later in Remark~\ref{rem:two-rounds}.
\begin{proof}
Let $C = \{G_1, \dots, G_m\}$ denote the $\BQP^\cO$ computation acting on input $\ket{\psi_0} = \ket{0^n}$, where each $G_i$ is either a single or two-qubit gate in any universal gate set $S$, or an oracle gate $\cO_f$:

\begin{equation*}
    {\cO_f} \ket{x} \rightarrow (-1)^{f(x)} \ket{x}.
\end{equation*}

Suppose that the state of the computation after applying gate $G_i$ is:

\begin{equation*}
    \ket{\psi_i} = \sum_{x \in \{0,1\}^n} \alpha_{i, x} \ket{x}
\end{equation*}

The proof would consist of claimed descriptions of all these states.
Suppose the claimed states by the prover are 
\begin{equation*}
    \ket{\tps_i} = \sum_{x \in \{0,1\}^n} \tilde{\alpha}_{i, x} \ket{x}
\end{equation*}
To verify $C$ it suffices to test that (i) the claimed states $\ket{\tps_i}$ are ``approximately truthful'' in the sense that they are close to the true states $\ket{\psi_i}$ and (ii) the probability of measuring $0$ on the first qubit of the last state is at least $2/3-o(1)$. 
To verify that the states are approximately truthful, it suffices to check that they are locally consistent, namely that for each gate $G_i$, 
    \begin{equation*}| 1 - \bra{\tps_i} G_i \ket{\tps_{i-1}}| < 1/\poly(n).
    \end{equation*}
    Indeed, as we will show later, such a condition implies that $\ket{\tps_i} \approx \ket{\psi_i}$ for all $i=0,\ldots, m$.

    To specify the states, the proof $\pi$ contains:\begin{enumerate}
        \item For every $i\in [m]$ and $x\in \{0,1\}^n$, a complex phase $\gamma_{i,x}\in \C$, which for the truthful proof would satisfy 
        \begin{equation}\label{eq:true_gamma}
            \gamma_{i,x} = \frac{\alpha_{i,x}}{|\alpha_{i,x}|}.
        \end{equation}
        \item  For every $i\in [m]$, $k\in [n]$, and prefix $w_1\dots w_k\in \{0,1\}^k$, a conditional probability $p_{i,w_k|w_1\dots w_k-1}\in [0,1]$ which for the truthful proof would satisfy\footnote{Of course, the verifier does not know whether \Cref{eq:true_gamma,eq:true_p_iw} hold.}  
        \begin{equation}\label{eq:true_p_iw} p_{i,w_k|w_1\dots w_k-1}  
        = \frac{\sum_{y\in \{0,1\}^{n-k}}|\alpha_{i,w_1\dots w_k y}|^2}{\sum_{y\in \{0,1\}^{n-(k-1)}}|\alpha_{i,w_1\dots w_{k-1} y}|^2}\end{equation}      
    \end{enumerate}

     We remark that without loss of generality we can assume that for any $i\in [m], k\in [n], w_1\ldots, w_{k-1}\in \{0,1\}$ it holds that $p_{i,0|w_1\dots w_{k-1}} + p_{i,1|w_1\dots w_{k-1}} = 1$, since the prover can just encode one of these conditional probabilities, say $p_{i,1|w_1\dots w_{k-1}}$, and the verifier would either read it directly if $w_k=1$ or compute $1-p_{i,1|w_1\dots w_{k-1}}$ if $w_k=0$. 









At a high level, $\pi = (\gamma, p)$ is the truth table of the state synthesis oracle from \cite{aaronson2016complexityquantumstatestransformations}. The idea is that each single (resp. two) qubit gate acts locally and thus can be verified by testing the 2 (resp. 4) dimensional subspace it acts on via local checks to $\pi$. Moreover, each oracle gate can be verified with a consistency check query to the oracle $\cO$. We now describe the verifier's algorithm, beginning with its subroutines. 

\paragraph{Computing State Amplitudes} The following subroutine describes the process that the verifier uses to recover the amplitudes $\tilde{\alpha}_{i, x}$, given a gate number $i \in [m]$ and an input $x \in \{0,1\}^n$.

\begin{algorithm}[H]\label{alg:compute-norms}
\caption{Computing State Amplitudes}
\KwIn{Query access to $\pi = (\gamma,p)$, index $i \in [m]$, input $x_1\dots x_n \in \{0,1\}^n$.}
\KwOut{$\tilde{\alpha}_{i, x} \in \C$}
\Return{$\gamma_{i,x} \cdot \left(\prod_{k=1}^{n} p_{i,x_k|x_1\dots x_{k-1}}\right)^{1/2}$}
\end{algorithm}











\begin{claim}\label{claim:correctness-amplitudes} For every proof $\pi \in \{0,1\}^*$ and $i \in [m]$, there exists a quantum state $\ket{\phi_i}$:
\begin{equation*}
    \ket{\phi_i} = \sum_{x \in \{0,1\}^n} \tilde{\alpha}_{i, x} \ket{x}
\end{equation*}
    such that Algorithm~\ref{alg:input-sampling} outputs $\tilde{\alpha}_{i, x}$ on input $\pi, i$ and $x \in \{0,1\}^n$.
\end{claim}

\begin{proof}
By our assumption that $p_{i,0|w_1\dots w_{k-1}} + p_{i,1|w_1\dots w_{k-1}} = 1$ we indeed get that for any $i\in [m]$ one can define a probability distribution over $\{0,1\}^n$ by taking $\Pr[X=x] = \prod_{k=1}^{n} p_{i,x_k|x_1\dots x_{k-1}}$.
As $|\gamma_{i,x}|=1$, this means that the sum of amplitudes squared equals $1$, and thus the state $\sum_{x \in \{0,1\}^n} \tilde{\alpha}_{i, x} \ket{x}$ is a valid quantum state.\end{proof}

\paragraph{Input Sampling} The verifier makes use of the following sampling subroutine. Given query access to the proof $\pi = (\gamma,p)$ and an initial index $i$, it samples $x \in \{0,1\}^n$ with probability $|\tilde{\alpha}_{i, x}|^2$. We write $wb$ to denote concatenation of string $w$ and bit $b$. We write $\varepsilon$ to denote the empty string.

\begin{algorithm}[H]\label{alg:input-sampling}
\caption{Conditional Input Sampling}
\KwIn{Query access to $\pi = (\gamma,p)$, index $i \in [m]$, input length $n \geq 1$.}
\KwOut{$x \in \{0,1\}^n$}
\BlankLine

Initialize $w = \varepsilon$.

\For{$k = 0$ \KwTo $n-1$}{
Query $p_{i, 1|w}$

Sample $b$ as a Bernoulli random variable with probability $p_{i,1|w}$ equaling $1$.

Update $w \gets wb$ 
}

\Return{$w$}

\end{algorithm}









\begin{claim}\label{claim:correctness-sampling} For every proof $\pi \in \{0,1\}^*$, $i \in [m]$ and $n \in \mathbb{N}$, let the state $\ket{\phi_i} = \sum_{x \in \{0,1\}^n} \tilde{\alpha}_{i, x} \ket{x}$
be the state in Claim~\ref{claim:correctness-amplitudes}. Then Algorithm~\ref{alg:input-sampling} outputs $x \in \{0,1\}^n$ with probability $|\tilde{\alpha}_{i, x}|^2$. 
\end{claim}

\begin{proof}
Let $X$ denote a random variable taking values in $\{0,1\}^n$. Algorithm~\ref{alg:input-sampling} outputs $X = x_1 \cdots x_n$ with the following probability: 
\begin{align*}
    \Pr[X = x_1 \cdots x_n] &= \prod_{k=1}^n \Pr[X_k = x_k | X_{1}\cdots X_{k-1} = x_1 \cdots x_{k-1}] 
    = \prod_{k = 1}^n p_{i,x_k|x_1\dots x_{k-1}}
     = \left|\tilde{\alpha}_{i, x}\right|^2
\end{align*}
The claim follows.
\end{proof}

\paragraph{Local Checks.} The following subroutine describes the verifier's queries to the proof $\pi$ to compute the coefficients used in the local checks of the protocol. We write $x \oplus e_q$ to denote the string obtained by flipping the $q^{th}$ bit of $x$ and $x_q$ to denote the $q^{th}$ bit of $x$.

\begin{algorithm}[h]\label{alg:local-checks}
\caption{Local Checks}
\KwIn{Query access to $\pi$, index $i \in [m]$, input $x \in \{0,1\}^n$, gate $G_i$.}
\KwOut{$\eta_{i, x} \in \mathbb{C}$}
\BlankLine
\begin{enumerate}

            \item If $G_i \in S$ and $G_i$ acts on a single qubit $q$ run Algorithm~\ref{alg:compute-norms} on: \[(\pi[i-1, x] , \pi[i-1, x \oplus e_q])\] to get \[(\tilde{\alpha}_{i-1, x}, \tilde{\alpha}_{i-1, x \oplus e_q}).\] 
            Compute
            \[ \eta_{i,x} = (G_{i})_{(x_q, x_q)} \cdot \tilde{\alpha}_{i-1,x} + (G_{i})_{(x_q, \bar{x}_q )} \cdot \tilde{\alpha}_{i-1,x \oplus e_q}. \]
            \item If $G_i \in S$ and $G_i$ acts on two qubits $(q, s)$, run Algorithm~\ref{alg:compute-norms} on: \[\pi[i-1, x], \pi[i-1, x \oplus e_q], \pi[i-1, x \oplus e_s], \pi[i-1, x \oplus e_q \oplus e_s]\] to get  \[(\tilde{\alpha}_{i-1, x}, \tilde{\alpha}_{i-1, x \oplus e_q}, \tilde{\alpha}_{i-1, x \oplus e_s}, \tilde{\alpha}_{i-1, x \oplus e_q \oplus e_s}).\] 
            
             Compute

             \begin{align*}
                \eta_{i,x} = (G_i)_{(x_qx_{s},  x_qx_s)} \cdot \tilde{\alpha}_{i-1, x} &+ (G_i)_{(x_qx_{s},  \bar{x}_qx_s)} \cdot \tilde{\alpha}_{i-1, x \oplus e_q} \\
                &+ (G_i)_{(x_qx_{s},  x_q\bar{x}_s)} \cdot \tilde{\alpha}_{i-1, x \oplus e_s} + (G_i)_{(x_qx_{s},  \bar{x}_q\bar{x}_s)} \cdot \tilde{\alpha}_{i-1, x \oplus e_q \oplus e_s}
             \end{align*}

        \item If $G_i = \cO_f$:

        run Algorithm~\ref{alg:compute-norms} on $\pi[i-1, x]$ to get $\tilde{\alpha}_{i-1, x}$ and query $\cO_f$ at $x$ to obtain $f(x)$. 
        
        Compute
            \[ \eta_{i,x} = (-1)^{f(x)} \cdot \tilde{\alpha}_{i-1,x}. \]
        
       
        \end{enumerate}
\Return{$\eta_{i, x}$}

\end{algorithm}

\begin{claim}\label{claim:correctness-local-checks}
   Let $\ket{\phi_{i-1}} \in \cH^{2^n}$ be the state from Claim~\ref{claim:correctness-amplitudes}. For any $x \in \{0,1\}^n$ Algorithm~\ref{alg:local-checks} outputs:

   \begin{equation*}
       \eta_{i, x} = \bra{x}G_i \ket{\phi_{i-1}}
   \end{equation*}
\end{claim}

\begin{proof}
The state $\ket{\phi_{i-1}}$ has the following form: 

\begin{equation*}
    \ket{\phi_{i-1}} = \sum_{x \in \{0,1\}^n} \tilde{\alpha}_{i-1, x} \ket{x}
\end{equation*}

We analyze each case of $G_i$ separately. 

\paragraph{Single-qubit gates} If $G_i \in S$ and $G_i$ acts on a single qubit $q$, then $G_i$ splits $\cH^{2^n}$ into $2^{n-1}$ subspaces, each spanned by:

\begin{equation*}
    \ket{x_1, \dots, x_{q-1}, 0, x_{q+1}, \dots x_n} \text{ and } \ket{x_1, \dots, x_{q-1}, 1, x_{q+1}, \dots x_n},
\end{equation*}

for a fixed $x_1, \dots, x_n \in \{0,1\}^n$. Therefore, for $x_1, \dots, x_n \in \{0,1\}^n$, $\tilde{\alpha}_{i-1, x}$ is related to the new amplitude $\tilde{\alpha}_{i, x}$ by: 

\begin{equation*}
    \tilde{\alpha}_{i, x} = G_{x_q0} \cdot \tilde{\alpha}_{i-1, x_1, x_2, \dots, x_{q-1}, 0, x_{q+1}, \dots, x_n} + G_{x_q1} \cdot \tilde{\alpha}_{i-1, x_1, x_2, \dots, x_{q-1}, 1, x_{q+1}, \dots, x_n},
\end{equation*}

which are the checks performed in the protocol.

\paragraph{Two-qubit gates} If $G_i \in S$ and $G_i$ acts on two qubits $q, s$, then $G_i$ splits $\cH^{2^n}$ into $2^{n-2}$ subspaces, each spanned by the following vectors, for any $x_1, \dots, x_n \in \{0,1\}^n$:

\begin{align}
     \ket{x_1, \dots, x_{q-1}, 0, x_{q+1}, \dots, x_{s-1}, 0, x_{s+1}, \dots, x_n}, \nonumber\\
    \ket{x_1, \dots, x_{q-1}, 0, x_{q+1}, \dots, x_{s-1}, 1, x_{s+1}, \dots, x_n}, \nonumber\\
     \ket{x_1, \dots, x_{q-1}, 1, x_{q+1}, \dots, x_{s-1}, 0, x_{s+1}, \dots, x_n},\nonumber\\
      \ket{x_1, \dots, x_{q-1}, 1, x_{q+1}, \dots, x_{s-1}, 1, x_{s+1}, \dots, x_n}.\nonumber
\end{align}

Therefore, for a fixed $x_1, \dots, x_n \in \{0,1\}^n$, $\tilde{\alpha}_{i-1, x}$ is related to the new amplitude $\tilde{\alpha}_{i, x}$ by: 

\begin{align}
    \tilde{\alpha}_{i, x} = G_{(x_qx_s, 00)} \cdot \tilde{\alpha}_{i-1, x_1, \dots, x_{q-1}, 0, x_{q+1}, \dots, x_{s-1}, 0, x_{s+1}, \dots, x_n} \nonumber \\
    + G_{(x_qx_s, 01)} \cdot \tilde{\alpha}_{i-1, x_1, \dots, x_{q-1}, 0, x_{q+1}, \dots, x_{s-1}, 1, x_{s+1}, \dots, x_n} \nonumber \\
    + G_{(x_qx_s, 10)} \cdot \tilde{\alpha}_{i-1, x_1, \dots, x_{q-1}, 1, x_{q+1}, \dots, x_{s-1}, 0, x_{s+1}, \dots, x_n} \nonumber\\
    + G_{(x_qx_s, 11)} \cdot \tilde{\alpha}_{i-1, x_1, \dots, x_{q-1}, 1, x_{q+1}, \dots, x_{s-1}, 1, x_{s+1}, \dots, x_n}\nonumber
\end{align}

These are exactly the checks performed in the protocol.

\paragraph{Oracle gates} If $G_i = \cO_f$:

\begin{equation}\label{eq:desired-outcome}
    \bra{x}G_i\ket{\phi_{i-1}} = (-1)^{f(x)} \cdot \tilde{\alpha}_{i-1, x}. 
\end{equation}

Queries to the proof $\pi$ return $\tilde{\alpha}_{i-1, x}$
and the algorithm sets:

\begin{equation*}
    \eta_{i, x} = (-1)^{f(x)} \cdot \tilde{\alpha}_{i-1, x},
\end{equation*}

which matches Equation~\ref{eq:desired-outcome}.

\end{proof}

\paragraph{Full Verifier Algorithm} We now describe the full algorithm of the verifier.

\begin{algorithm}[H]\label{alg:verifier-alg}
\caption{Full Verifier Algorithm}
\KwIn{Query access to $\pi$.}
\KwOut{$\{\acc, \rej\}$}
\BlankLine

\begin{enumerate}
        \item Query $\tilde{\alpha}_{0, 0^n} \leftarrow \pi[0, 0^n]$ using Algorithm~\ref{alg:compute-norms}. Output $\rej$ if $1 \neq \tilde{\alpha}_{0, 0^n}$.
        \item For every $i \in \{1, \dots, m\}$:
        
        \begin{enumerate}
        \item If $i = m$: query $p_{m,1|\eps}$ and output $\rej$ if \[p_{m,1|\eps} < 2/3.\]
        \item Repeat the following $t = \poly(n)$ times:
        \begin{enumerate}
        \item Sample $x \leftarrow \{0,1\}^n$ using Algorithm~\ref{alg:input-sampling}.     
        \item Query $\tilde{\alpha}_{i, x} \leftarrow \pi[i, x]$ using Algorithm~\ref{alg:compute-norms}. 
        

        \item Use Algorithm ~\ref{alg:local-checks} to compute $\eta_{i, x}$.
        

        \item Output $\rej$ if \begin{equation*}|\eta_{i,x} - \tilde{\alpha}_{i,x}| \ge 2^{-100n}.\end{equation*}
        \end{enumerate}

        
\end{enumerate} 
        \item Output $\acc$.
    \end{enumerate}

\end{algorithm}

The soundness of this protocol relies on the following lemma and its corollary.

\begin{lemma}
\label{lem:inner-product-estimation}
Let, $\ket{\psi}, \ket{\phi}$ be normalized vectors in $\mathbb{C}^{2^n}$.
The random variable
\[\gamma = \frac{1}{k} \sum_{i=1}^k \frac{\phi_{X_i}}{\psi_{X_i}},\] where ${X_1}, \ldots, {X_k}$ are $k$ independent samples from the distribution $\Pr[X=x] = \abs{\psi_x}^2$, is an unbiased estimator for $\ip{\psi}{\phi}$ such that, for any $\delta > 0$,

\begin{equation*}
    \Pr_{X \sim \psi}\left[|\gamma - \ip{\psi}{\phi}| \geq \delta \right] \leq \frac{1}{k \delta^2}.
\end{equation*}

\end{lemma}

\begin{proof}
Observe that:
\begin{align}
  \ip{\psi}{\phi}
  &= \sum_x \psi_x^\ast \phi_x \label{eq:inner-prod-sum} \\
  &= \sum_x \abs{\psi_x}^2 \cdot \frac{\phi_x}{\psi_x} \nonumber \\
  &= \E_{X \sim \psi}\!\left[\frac{\phi_x}{\psi_x}\right],
  \label{eq:inner-prod-expectation}
\end{align}
where $X$ is a random variable taking values in $\{0,1\}^n$ with
$\Pr[X=x]=\abs{\psi_x}^2$ and $\psi_x\neq 0$ whenever we sample $x$.\footnote{For
indices with $\psi_x=0$ we may define $\frac{\phi_x}{\psi_x} := 0$ since they are never
sampled.} In what follows, we drop the specification of the underlying distribution $X \sim \psi$ for clarity. 

Next we upper bound the variance of $\frac{\phi_x}{\psi_x}$.  We have:
\begin{align*}
  \mathbb{E}\left[\left|\frac{\phi_x}{\psi_x}\right|^2\right]
  &= \sum_x \Pr[X=x] \cdot \frac{\abs{\phi_x}^2}{\abs{\psi_x}^2}
   = \sum_x \abs{\psi_x}^2 \cdot \frac{\abs{\phi_x}^2}{\abs{\psi_x}^2}
   = \sum_x \abs{\phi_x}^2
   = 1,
\end{align*}

Therefore:
\[
  \var\left(\frac{\phi_x}{\psi_x}\right)
  = \mathbb{E}\left[\abs{\frac{\phi_x}{\psi_x}}^2\right] - \abs{\E\left[\frac{\phi_x}{\psi_x}\right]}^2
  = 1 - \abs{\ip{\psi}{\phi}}^2
  \le 1.
\]

Now take $k$ independent samples $X_1,\dots,X_k$ from the distribution
$\Pr[X=x]=\abs{\psi_x}^2$ and define the empirical mean:
\[
  \gamma
  := \frac{1}{k} \sum_{i=1}^k \frac{\phi_{X_i}}{\psi_{X_i}}.
\]

Then $\gamma$ is an unbiased estimator for $\ip{\psi}{\phi}$:
\[
  \E[\gamma]
  = \frac{1}{k} \sum_{i=1}^k \mathbb{E}\left[\frac{\phi_{X_i}}{\psi_{X_i}}\right] 
  = \frac{1}{k} \cdot k \cdot \mathbb{E}\left[\frac{\phi_{X}}{\psi_{X}}\right] 
  = \ip{\psi}{\phi}.
\]
Since the samples are independent,
\[
  \var(\gamma)
  = \frac{1}{k^2}\sum_{i=1}^k \var\!\left(\frac{\phi_{X_i}}{\psi_{X_i}}\right)
  = \frac{1}{k^2}\cdot k \cdot \var\left(\frac{\phi_{X}}{\psi_{X}}\right)
  = \frac{\var\left(\frac{\phi_{X}}{\psi_{X}}\right)}{k}
  \le \frac{1}{k}.
\]

Now, applying Chebyshev:

\[
  \Pr\left[\,\abs{\gamma - \ip{\psi}{\phi}} \geq \epsilon \right]
  \le \frac{\var{(\gamma)}}{\epsilon^2} \leq \frac{1}{k \cdot \epsilon^2}.
\]
\end{proof}

\begin{corollary}\label{cor:approx_eq}
Let, $\ket{\psi}, \ket{\phi}$ be normalized vectors in $\mathbb{C}^{2^n}$. Let $2^{-n} <\delta \le 1$.
Suppose $|1-\ip{\psi}{\phi}| > \delta$.
Then, with probability at least $\delta^2/10$ over $X\sim \psi$ it holds that 
$|\phi_X-\psi_X| \ge 2^{-100n}$.
\end{corollary}
\begin{proof}
    Suppose by contradiction that with probability less than $\delta^2/10$ over $X\sim \psi$, it holds that  $|\phi_X-\psi_X| \ge 2^{-100n}$. Then by union bound on $k = 4/\delta^2$ samples, with probability at least $1/2$ we have that all $X_1,\ldots, X_{k}$ satisfy $|\phi_{X_i}-\psi_{X_i}|\le 2^{-100n}$ and also samples satisfy $|\psi_{X_i}|^2\ge 2^{-50n}$, for all $i\in [k]$. But then, \[|1-\psi_{X_i}/\phi_{X_i}|\le 2^{-50n}\]
    for all $i\in [k]$, and the empirical mean $\gamma$ in \Cref{lem:inner-product-estimation} would be exponentially close to $1$. But this means that $|\gamma-\ip{\psi}{\phi}|\ge \delta-2^{-50n}\ge 0.9\delta$.
    We get a contradiction as \[\frac12 \le \Pr[|\gamma-\ip{\psi}{\phi}|\ge 0.9\delta] \le \frac{1}{k (0.9\delta)^2} = \frac1{(4/\delta^2) (0.9\delta)^2} < \frac{1}{3}\;.\]
\end{proof}


\paragraph{Completeness.}

\begin{claim}
    Let $\pi$ be the proof described above and let $V^{\pi, \cO}$ be the algorithm defined in Algorithm~\ref{alg:verifier-alg}. If $x \in \YES$, then: 

    \begin{equation*}
        \Pr[V^{\pi, \cO}(x) = \acc] = 1.
    \end{equation*}

\end{claim}

\begin{proof}
\avishay{Technically, we should talk about the errors due to precision more accurately. I attempted to do this}
    For $\pi$ provided by an honest prover, $\tilde{\alpha}_{0, 0^n} = 1$ and thus step $1$ of the protocol will pass with probability $1$. All entries of the proof are recorded with $\poly(n)$ bits of precision, and therefore an honest prover will respond with precision $1/2^{\poly(n)}$. 
    If $i = m$, and $x \in L^\cO$, the proof will be accepted with probability at least  $2/3$, and thus the verifier will pass check $2a$. 
    For each $i \in [m]$, the verifier checks whether $\tilde{\alpha}_{i, x}$ and $\eta_{i, x}$ are the same up to error $2^{-\poly(n)}$. Indeed, in the honest proof, $\tilde{\alpha}_{i, x}$ and $\eta_{i, x}$ would be the same up to an additive error of at most $2^{-\poly(n)}$ due to precision, so step 2b(iv) would never reject.
%
\end{proof}



\paragraph{Soundness.} 

\begin{claim}
    Let $V^{\pi, \cO}$ be the algorithm defined in Algorithm~\ref{alg:verifier-alg}. If $x \in \NO$, then, for all $\pi \in \{0,1\}^*$:

    \begin{equation*}
        \Pr[V^{\pi, \cO}(x) = \rej] \geq 1- 2^{-\poly(n)}
    \end{equation*}
\end{claim}

\begin{proof}
If the initial state corresponding to the proof $\pi$ is not equal to $\ket{0^n}$, then the verifier will output $\rej$ in step $1$ with probability $1$. If the final state corresponding to the proof $\pi$ has probability of measuring $0$ in the first bit being $< 2/3$, then the verifier will output $\rej$ with probability $1$ in step $2(a)$. For the rest of the proof, we assume that the initial state is equal to $\ket{0^n}$ and that the final state has $0$ in the first bit with probability $> 2/3$. We will use this to show that the proof was dishonest during a propagation term, since we know that $x \in \NO$ and thus these two conditions cannot both hold for a valid computation history.

Denote by $\ket{\tilde{\psi}_i}$ the state corresponding to the proof $\pi$ after applying gate $G_i$. Since $x \in \NO$ and, by assumption, the initial state is $\ket{0^n}$ and the final state outputs $0$ with probability $\geq 2/3$, there must exist an $i \in [m]$ such that:

\begin{equation}\label{eq:inner-prod-fail}
    | 1 - \bra{\tilde{\psi}_i} G_i \ket{\tilde{\psi}_{i-1}} | > \delta
\end{equation}
for $\delta = \frac{1}{\poly(n)}$.

We show this by contradiction. Suppose that, for all $i$:

\begin{equation}\label{eq:locally_close}
    | 1 - \bra{\tilde{\psi}_i} G_i \ket{\tilde{\psi}_{i-1}} | \leq \delta.
\end{equation}
\avishay{Edited the following (until the end of the proof)}
Recall that the true states of the BQP algorithm are given by $\ket{\psi_0}, \ldots, \ket{\psi_m}$ where $\ket{\psi_0} = \ket{0^n}$ and for $i\in [m]$, $\ket{\psi_i} = G_i \ket{\psi_{i-1}}$.
We show by induction on $i\in \{0,1\ldots, m\}$ that 
$\| \ket{\tps_i} - \ket{\psi_i} \|\le i \cdot \sqrt{2\delta}$.
The base case is trivially true as both vectors $\ket{\tps_0}$ and $\ket{\psi_0}$ are equal to $\ket{0^n}$.
Next, for $i>0$, we note that by \Cref{eq:locally_close}, 
\begin{align}
    \|\ket{\tps_i} - G_i \ket{\tps_{i-1}}\|^2 &= 2 - 2\Re \bra{\tps_i} G_i \ket{\tps_{i-1}} \\ &\leq 2\delta \\
    \| \ket{\tps_i} - G_i \ket{\tps_{i-1}} \| &\leq \sqrt{2\delta}.\label{eq:71}\end{align}
Thus,
\begin{align*}
    \| \ket{\tps_i} - \ket{\psi_i} \|  &= 
    \| \ket{\tps_i} - G_i \ket{\tps_{i-1}} + G_i \ket{\tps_{i-1}} - \ket{\psi_i}\| \\
    &\le 
    \| \ket{\tps_i} - G_i \ket{\tps_{i-1}}\| + \|G_i \ket{\tps_{i-1}} - \ket{\psi_i}\|  
    \tag{By triangle inequality} \\
    &\le 
     \sqrt{2\delta} + \|G_i \ket{\tps_{i-1}} - G_i \ket{\psi_{i-1}}\|
     \tag{\Cref{eq:71}}\\
     &=
     \sqrt{2\delta} + \|\ket{\tps_{i-1}} - \ket{\psi_{i-1}}\|
     \tag{$G_i$ is a unitary}\\
     &\le \sqrt{2\delta} + (i-1)\sqrt{2\delta} 
     \tag{Induction hypothesis}\\&= i\sqrt{2\delta}
\end{align*}
Overall, we get that for \begin{align*}
     \ket{u} := \ket{\psi_m}  = G_m G_{m-1} \dots G_1 \ket{0^n}, \qquad \ket{v} &:= \ket{\tps_m}\end{align*}
    it holds that 
    \begin{align*}
    \| \ket{u} - \ket{v} \| &\leq m \sqrt{2\delta}.
\end{align*}



Since $x \in \NO$, this implies that the probability that measuring $\ket{\psi_m}$ and obtaining $0$ on the first qubit is $\leq 1/3$. It follows that the probability of measuring $\ket{\tps_m}$ and obtaining $0$ on the first qubit is at most
$\Pr[1] \leq 1/3 + 2m \sqrt{2\delta}$.

\begin{align*}
    \bra{u} \Pi_0\ket{u} &\leq 1/3 \\
    \bra{v} \Pi_0 \ket{v} &= \bra{u}\Pi_0\ket{u} + (\bra{v} - \bra{u}) \Pi_0 \ket{u} + \bra{v} \Pi_0 (\ket{v} - \ket{u}) \\
    &\leq 1/3 + m\sqrt{2\delta} + m\sqrt{2\delta}.
\end{align*}

If we set $\delta \ll 1/m^2$, then we get that $\bra{v} \Pi_0 \ket{v} < 2/3$, which contradicts the assumption that $\bra{\tps_m} \Pi_0 \ket{\tps_m} \geq 2/3$. Therefore, there must exist an $i$ such that \eqref{eq:inner-prod-fail} holds. 

We are now going to argue that the verifier must reject with high probability on round $i$. Indeed, by picking $t = \poly(n)/\delta^2$ we get that the verifier would  catch the prover's inconsistency in round $i$ except with extremely small failure probability  \[\Pr[\text{pass $t$ checks in round $i$}] \le (1-\delta^2/10)^t\le 2^{-\poly(n)}\]
by Corollary~\ref{cor:approx_eq}.

\end{proof}



\end{proof}

\section{Compiling to $\MIP$}\label{sec:mip}

We use the clause-variable game to compile the $\PCP$ protocol from above into an $\MIP$.

\begin{theorem}\label{thm:MIP_for_BQP}
    For any oracle $\cO$:
    \begin{equation}
        \BQP^\cO \subseteq \MIP^\cO
    \end{equation}
\end{theorem}

The proof will be based on a reduction of any adaptive PCP relative to an oracle to an MIP relative to the same oracle.

\begin{lemma}\label{lemma:adPCP_to_2_round_MIP}
    For any oracle $\cO$, if $L\in \mathsf{adPCP}^{\cO}_{s,1}[p(n), q(n),r(n)]$, then $L$ has a $\MIP^{\cO}$  protocol with two adaptive rounds of communication between the verifier and prover, and with
    \begin{itemize}
        \item perfect completeness,
        \item soundness $1-(1-s)/p(n)$,
        \item $r(n) + \log(p(n))$ randomness,
        \item $O(r(n) +p(n))$ bits of communication,
        \item $q(n)$ queries to $\cO$,
        \item where the verifier running time is $\poly(n, p(n),q(n),r(n))$.
    \end{itemize}
\end{lemma}

\begin{proof}
Let $A = A^{\cO,\pi}$ be the adaptive PCP verifier for $L$.
We describe an $\MIP^{\cO}$ protocol for $L$.  
Let $N_{\pi} \le p(n)\cdot 2^{r(n)}$ be the proof length of the adaptive PCP on input $x$.

\paragraph{The protocol on input $x$:}
    \begin{enumerate}
        \item The verifier picks a random string $R\in \{0,1\}^{r(n)}$ and sends $R$ to Prover~1.
        \item Prover~1 sends $p(n)$ many bits in response $a = (a_1, \ldots, a_{p(n)})$.
        \item 
        The verifier simulates the adaptive PCP verifier, $A$, on randomness $R$.
        Whenever $A$ asks to query $\pi$ we feed it with the next bit the sequence $a$. Whenever $A$ asks to query $\cO$, we perform the query as is.
        If finally $A$ decides to reject, then immediately reject!
        \item Let $i_1, \ldots, i_{p(n)}\in [N_{\pi}]$ be the proof locations ``queried'' by the verifier in the above simulation.
        \item The verifier picks $j\in [p(n)]$ at random and sends $i_j$ to Prover~2.
        \item Prover~2 responds with $b\in \{0,1\}$.
        \item The verifier accepts if and only if $a_j = b$.
    \end{enumerate}

See \Cref{fig:mip-colored} for a visualization of the protocol.

\paragraph{Analysis:}

    We observe that the protocol indeed uses $r(n)+\log(p(n))$ bits of randomness,  sends $(r(n) + p(n) +\log(N_{\pi}) +1) = O(r(n) + p(n))$ bits of communication, and makes at most $q(n)$ queries to $\cO$.

        \medskip\noindent\textbf{Completeness.}
Next, we argue about \textbf{perfect completeness}. Suppose $x\in L$.
    Then by the assumption on $L$, there exists a proof $\pi$ such that $A^{\cO, \pi}$ always accepts. Then, if Prover~1 and Prover~2 answer truthfully according to the proof $\pi$, then the check in Step~3 will always pass as the simulation of  $A$ accepts, and so does the consistency check in Step~7.

    \medskip\noindent\textbf{Soundness.}
    Let $x\notin L$. By the assumption on the soundness of $A$, we know that for any proof $\pi'$, $A^{\pi',\cO}$ rejects $x$ with probability at least $1-s$.
    Without loss of generality, the strategies of Prover~1 and Prover~2 are deterministic given the questions they get from the Verifier. Define $\pi'$ to be the answer strategy of Prover~2. That is, let $\pi'[\ell]$ be the response of Prover~2 if the verifier sends $\ell \in [N_{\pi}]$ on Step~5.

    By the assumption, we know that for at least $(1-s)$ fraction of choices of $R\in \{0,1\}^{r(n)}$ we have that $A^{\pi',\cO}$ on randomness $R$ would reject. Call any such string $R$ a ``good'' string. Now, we show that if the verifier chooses a good string $R$, then she will reject with probability at least $1/p(n)$.
    Indeed, if Prover~1 answers according to $\pi'$, then the simulation would lead to reject immediately in Step~3.
    Otherwise, Prover~1 answers $a = (a_1, \ldots, a_{p(n)})$ where there is a first location $a_j$ where the answer differs from that of $\pi'$.
    Since the answers up to the $j$th query were according to $\pi'$, $i_j$ is calculated correctly, and if the verifier were to pick $j$ in Step~5, she would notice the inconsistency.

    This shows that the verifier would reject with probability at least $(1-s)/p(n)$.
\end{proof}

\begin{remark}\label{rem:two-rounds}
Any $2$-round $\MIP^{\cO}$ protocol can be converted back into a $2$-round $\adPCP^{\cO}$ with the same completeness, soundness, randomness, and query complexity by the standard $\MIP$-to-$\PCP$ construction: the $\adPCP$ proof will consist of the full strategy table of the provers, and the $\adPCP$ verifier will simulate the $\MIP$ verifier.
\end{remark}

\begin{corollary}\label{cor:MIP}
    For any oracle $\cO$, if $L\in \mathsf{adPCP}^{\cO}_{s,1}[p(n), q(n),r(n)]$, then for any $t\in \mathbb{N}$, $L$ has a $2t$-round $\MIP^{\cO}$ protocol
    with
    \begin{itemize}
        \item perfect completeness,
        \item soundness $(1-(1-s)/p(n))^t$,
        \item $(r(n) + \log(p(n)))\cdot t$ randomness,
        \item $O(r(n) +p(n))\cdot t$ communication, 
        \item $q(n)\cdot t$ queries to $\cO$,
        \item where the verifier running time is $\poly(n, p(n),q(n),r(n))\cdot t$.
    \end{itemize}
    In particular, picking $t = O(p(n)/(1-s))$ gives soundness at most $1/2$.
\end{corollary}
\begin{proof}
    Repeat the protocol from \Cref{lemma:adPCP_to_2_round_MIP} $t$ times sequentially, using fresh randomness for each instantiation.
\end{proof}

\begin{proof}[Proof of \Cref{thm:MIP_for_BQP}]
By Theorem~\ref{thm:pcp-for-bqp}, we get that for any $L\in \BQP^{\cO}$, it holds that $L \in \mathsf{adPCP}_{1/2,1}[p(n), q(n), r(n)]$, where $p(n),q(n),r(n) = \poly(n)$.
But then, by \Cref{cor:MIP}, for $t = O(p(n))$, we have an $\MIP^{\cO}$ protocol for $L$, where the verifier running time is $\poly(n, p(n),q(n),r(n))\cdot t = \poly(n)$.
\end{proof}

\ifanon

\section*{Acknowledgments}
Acknowledgments removed for anonymity. We used Claude (Opus 4.6 Extended) and ChatGPT (Extended Thinking) to assist with preparing figures and understanding prior work. 
\else

\section*{Acknowledgments}

We thank Umesh Vazirani, Thomas Vidick, and Orr Paradise for helpful conversations and feedback on this work. A large portion of this research was done while AN and AV were visiting the Simons Institute at UC Berkeley for the 2025 Summer Cluster on Quantum Computing. We thank the workshop organizers and the Simons Institute for the opportunity and their generous hospitality. SA acknowledges support from the DOE grant DE-SC0025615. AN acknowledges support from the National Science Foundation CAREER Award CCF-2339948. 
AT acknowledges support from the National Science Foundation CAREER Award CCF-2145474. AV acknowledges support from the National Science Foundation Graduate Research Fellowship 2141064. 
We used Claude (Opus 4.6 Extended) and ChatGPT (Extended Thinking) to assist with preparing figures and understanding prior work. 

\fi

\appendix

\section{Quantum Circuit for Forrelation}\label{app:forr}

We recall the standard single-layer forrelation circuit and its intermediate
states. Throughout, we let $H$ denote the Hadamard gate on a single qubit and
write $H^{\otimes n}$ for the $n$-fold tensor product. We denote the identity operator by $I$.

\begin{enumerate}
  \item \textbf{Initialization.}
  Start in the all-zero state
  \[
    \ket{\phi_0} := \ket{0^n}.
  \]

  \item \textbf{Create a uniform superposition.}
  Apply $H^{\otimes n}$, qubit-by-qubit:
  \begin{align}
        \ket{\phi_1} &= (H_1\otimes I^{\otimes n-1})\ket{\phi_0},\\
        \ket{\phi_2} &= (I\otimes H_2\otimes I^{\otimes n-2})\ket{\phi_1},\\
        &\vdots \quad \vdots \nonumber\\
        \ket{\phi_{n}} &= I^{\otimes n-1} \otimes H_n \ket{\phi_{n-1}} \\
        &= \frac{1}{\sqrt{N}} \sum_{x \in \{0,1\}^n} \ket{x}
  \end{align}

  \item \textbf{Query $\cO_f$.}
  \begin{align}
    \ket{\phi_{n+1}}
    &:= \cO_f \ket{\phi_n} \\
    &= \frac{1}{\sqrt{N}} \sum_{x \in \{0,1\}^n} f(x)\ket{x}.
  \end{align}

  \item \textbf{Apply Hadamards again, qubit-by-qubit.}

  \begin{align}
        \ket{\phi_{n+2}} &= (H_1\otimes I^{\otimes n-1})\ket{\phi_{n+1}},\\
        \ket{\phi_{n+3}} &= (I\otimes H_2\otimes I^{\otimes n-2})\ket{\phi_{n+2}},\\
        &\vdots \quad \vdots \nonumber\\
        \ket{\phi_{2n+1}} &= I^{\otimes n-1} \otimes H_n \ket{\phi_{2n}} \\
        &= \frac{1}{N} \sum_{x,y\in\{0,1\}^n}
       f(x)(-1)^{x\cdot y}\ket{y}.
  \end{align}

  \item \textbf{Query $\cO_g$.}
  \begin{align}
    \ket{\phi_{2n+2}}
    &:= \cO_g \ket{\phi_{2n+1}} \\
    &= \frac{1}{N} \sum_{x,y\in\{0,1\}^n}
       f(x)(-1)^{x\cdot y}g(y)\ket{y}.
  \end{align}

  \item \textbf{Apply Hadamards a final time, qubit-by-qubit.}

  \begin{align}
        \ket{\phi_{2n+3}} &= (H_1\otimes I^{\otimes n-1})\ket{\phi_{2n+2}},\\
        \ket{\phi_{2n+4}} &= (I\otimes H_2\otimes I^{\otimes n-2})\ket{\phi_{2n+3}},\\
        &\vdots \quad \vdots \nonumber\\
        \ket{\phi_{3n+2}} &= I^{\otimes n-1} \otimes H_n \ket{\phi_{3n+1}} \\
        &= \frac{1}{N^{3/2}}
       \sum_{x,y,z\in\{0,1\}^n}
       f(x)(-1)^{x\cdot y}g(y)(-1)^{y\cdot z}\ket{z}.
  \end{align}
\end{enumerate}
The amplitude of $\ket{0^n}$ in the final state $\ket{\phi_{3n+2}}$ is
\begin{align}
  \ip{0^n}{\phi_{3n+2}}
  &= \frac{1}{N^{3/2}}
     \sum_{x,y} f(x)(-1)^{x\cdot y}g(y),
\end{align}
which is exactly $\ip{f}{\hat{g}}$ as in \Cref{eq:forr-phi}.
Therefore the probability of measuring $\ket{0^n}$ is
\begin{equation}
  p(f_1, f_2) := \Pr[0^n]
  = \abs{\ip{0^n}{\phi_{3n+2}}}^2
  = \frac{1}{N^3}
    \left(
      \sum_{x,y} f(x)(-1)^{x\cdot y}g(y)
    \right)^2.
\end{equation}
Thus estimating $p(f,g)$ (up to inverse-polynomial accuracy) is equivalent to
estimating $\ip{f}{\hat{g}}$ and hence to solving the Forrelation problem in Problem~\ref{pr:forrelation}.

\addcontentsline{toc}{section}{References}
\bibliographystyle{myhalpha}
\bibliography{quantum.bib}

\end{document}

%% file: figures/forrelation-circuit.tex
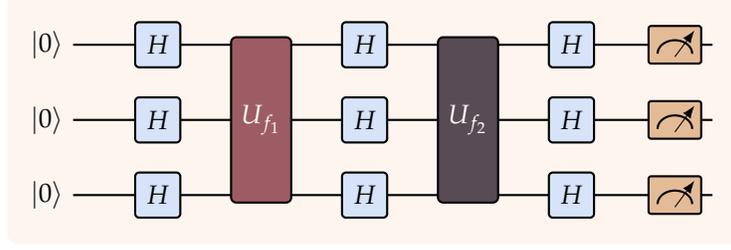
\begin{figure}[h]
\centering
\definecolor{bgcol}{HTML}{FEF5EF}
\definecolor{hadcol}{HTML}{D6E3F8}
\definecolor{ufcol}{HTML}{9D5C63}
\definecolor{ugcol}{HTML}{584B53}
\definecolor{meascol}{HTML}{E4BB97}
\begin{tikzpicture}[
    x=1.25cm,y=1cm,thick,>=Latex,
    had/.style={draw, minimum width=6mm, minimum height=6mm, fill=hadcol, rounded corners=1.5pt},
    meas/.style={draw, minimum width=7mm, minimum height=5mm, fill=meascol, rounded corners=1pt},
    oracle/.style={draw, minimum width=8mm, minimum height=2.2cm, rounded corners=2pt}
]
\fill[bgcol, rounded corners=4pt] (-0.7,-0.65) rectangle (7.1,2.65);
\foreach \y in {2,1,0}{
  \draw (0,\y) -- (6.8,\y);
  \node[left] at (0,\y) {$\ket{0}$};
}
\foreach \y in {2,1,0}{
  \node[had] at (0.9,\y) {$H$};
}
\node[oracle, fill=ufcol] at (2.0,1) {\textcolor{white}{$U_{f_1}$}};
\foreach \y in {2,1,0}{
  \node[had] at (3.1,\y) {$H$};
}
\node[oracle, fill=ugcol] at (4.2,1) {\textcolor{white}{$U_{f_2}$}};
\foreach \y in {2,1,0}{
  \node[had] at (5.3,\y) {$H$};
}
\foreach \y in {2,1,0}{
  \node[meas] (m\y) at (6.4,\y) {};
  \draw ([xshift=-2.3mm, yshift=-1.2mm]m\y.center)
    arc[start angle=180, end angle=0, x radius=2.3mm, y radius=1.8mm];
  \draw[->, thin] ([yshift=-1.6mm]m\y.center) -- ([xshift=2.6mm,yshift=1.9mm]m\y.center);
}
\end{tikzpicture}
\caption{The quantum circuit for 2-fold $\Forr$.}
\label{fig:two-fold-forrelation}
\end{figure}

%% file: figures/forrelation-circuit-sequential.tex
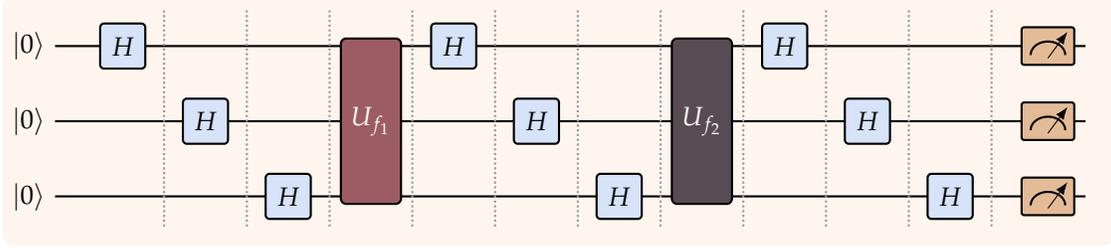
\begin{figure}[h]
\centering
\definecolor{bgcol}{HTML}{FEF5EF}
\definecolor{hadcol}{HTML}{D6E3F8}
\definecolor{ufcol}{HTML}{9D5C63}
\definecolor{ugcol}{HTML}{584B53}
\definecolor{meascol}{HTML}{E4BB97}
\begin{tikzpicture}[
    x=1.0cm,y=1cm,thick,>=Latex,
    had/.style={draw, minimum width=6mm, minimum height=6mm, fill=hadcol, rounded corners=1.5pt},
    meas/.style={draw, minimum width=7mm, minimum height=5mm, fill=meascol, rounded corners=1pt},
    oracle/.style={draw, minimum width=8mm, minimum height=2.2cm, rounded corners=2pt}
]

\fill[bgcol, rounded corners=4pt] (-0.7,-0.65) rectangle (14.1,2.65);

\foreach \y in {2,1,0}{
  \draw (0,\y) -- (13.7,\y);
  \node[left] at (0,\y) {$\ket{0}$};
}

\node[had] at (0.9,2) {$H$};
\draw[densely dotted, black!40] (1.45,-0.4) -- (1.45,2.5);

\node[had] at (2.0,1) {$H$};
\draw[densely dotted, black!40] (2.55,-0.4) -- (2.55,2.5);

\node[had] at (3.1,0) {$H$};
\draw[densely dotted, black!40] (3.65,-0.4) -- (3.65,2.5);

\node[oracle, fill=ufcol] at (4.2,1) {\textcolor{white}{$U_{f_1}$}};
\draw[densely dotted, black!40] (4.75,-0.4) -- (4.75,2.5);

\node[had] at (5.3,2) {$H$};
\draw[densely dotted, black!40] (5.85,-0.4) -- (5.85,2.5);

\node[had] at (6.4,1) {$H$};
\draw[densely dotted, black!40] (6.95,-0.4) -- (6.95,2.5);

\node[had] at (7.5,0) {$H$};
\draw[densely dotted, black!40] (8.05,-0.4) -- (8.05,2.5);

\node[oracle, fill=ugcol] at (8.6,1) {\textcolor{white}{$U_{f_2}$}};
\draw[densely dotted, black!40] (9.15,-0.4) -- (9.15,2.5);

\node[had] at (9.7,2) {$H$};
\draw[densely dotted, black!40] (10.25,-0.4) -- (10.25,2.5);

\node[had] at (10.8,1) {$H$};
\draw[densely dotted, black!40] (11.35,-0.4) -- (11.35,2.5);

\node[had] at (11.9,0) {$H$};
\draw[densely dotted, black!40] (12.45,-0.4) -- (12.45,2.5);

\foreach \y in {2,1,0}{
  \node[meas] (m\y) at (13.2,\y) {};
  \draw ([xshift=-2.3mm, yshift=-1.2mm]m\y.center)
    arc[start angle=180, end angle=0, x radius=2.3mm, y radius=1.8mm];
  \draw[->, thin] ([yshift=-1.6mm]m\y.center) -- ([xshift=2.6mm,yshift=1.9mm]m\y.center);
}

\end{tikzpicture}
\caption{The verification circuit for 2-fold $\Forr$ used in our protocol. Each layer of the original circuit is expanded sequentially to verify the circuit gate-by-gate.}
\label{fig:two-fold-forrelation-sequential}
\end{figure}

%% file: figures/prefix-tree2.tex
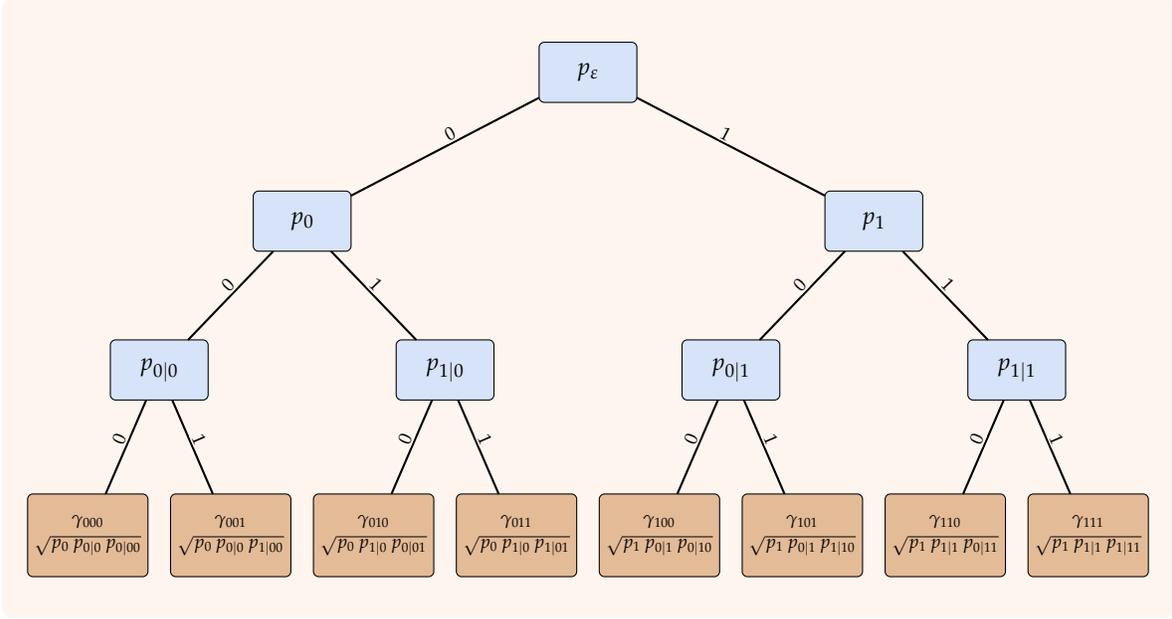
\begin{figure}[t]
\centering
\definecolor{bgcol}{HTML}{FEF5EF}
\definecolor{intcol}{HTML}{D6E3F8}
\definecolor{leafcol}{HTML}{E4BB97}
\begin{tikzpicture}[
  x=0.95cm,
  y=1.10cm,
  line cap=round,
  line join=round,
  >=Latex,
  every node/.style={font=\small},
  betanode/.style={
    draw,
    rounded corners=2pt,
    fill=intcol,
    minimum width=13mm,
    minimum height=8mm,
    inner sep=1.5pt,
    align=center
  },
  leafnode/.style={
    draw,
    rounded corners=2pt,
    fill=leafcol,
    minimum width=16mm,
    minimum height=11mm,
    inner sep=2pt,
    align=center,
    font=\scriptsize
  },
  bitlab/.style={
    font=\scriptsize,
    fill=bgcol,
    inner sep=0.7pt
  }
]
\fill[bgcol, rounded corners=4pt] (-8.2,-6.6) rectangle (8.2,0.9);
\node[betanode] (r) at (0,0) {$p_{\varepsilon}$};
\node[betanode] (n0) at (-4.0,-1.8) {$p_{0}$};
\node[betanode] (n1) at ( 4.0,-1.8) {$p_{1}$};
\node[betanode] (n00) at (-6.0,-3.6) {$p_{0|0}$};
\node[betanode] (n01) at (-2.0,-3.6) {$p_{1|0}$};
\node[betanode] (n10) at ( 2.0,-3.6) {$p_{0|1}$};
\node[betanode] (n11) at ( 6.0,-3.6) {$p_{1|1}$};
\node[leafnode] (l000) at (-7.0,-5.6) {$\gamma_{000}$\\[1pt]
$\sqrt{p_{0}\, p_{0|0}\, p_{0|00}}$};
\node[leafnode] (l001) at (-5.0,-5.6) {$\gamma_{001}$\\[1pt]
$\sqrt{p_{0}\, p_{0|0}\, p_{1|00}}$};
\node[leafnode] (l010) at (-3.0,-5.6) {$\gamma_{010}$\\[1pt]
$\sqrt{p_{0}\, p_{1|0}\, p_{0|01}}$};
\node[leafnode] (l011) at (-1.0,-5.6) {$\gamma_{011}$\\[1pt]
$\sqrt{p_{0}\, p_{1|0}\, p_{1|01}}$};
\node[leafnode] (l100) at ( 1.0,-5.6) {$\gamma_{100}$\\[1pt]
$\sqrt{p_{1}\, p_{0|1}\, p_{0|10}}$};
\node[leafnode] (l101) at ( 3.0,-5.6) {$\gamma_{101}$\\[1pt]
$\sqrt{p_{1}\, p_{0|1}\, p_{1|10}}$};
\node[leafnode] (l110) at ( 5.0,-5.6) {$\gamma_{110}$\\[1pt]
$\sqrt{p_{1}\, p_{1|1}\, p_{0|11}}$};
\node[leafnode] (l111) at ( 7.0,-5.6) {$\gamma_{111}$\\[1pt]
$\sqrt{p_{1}\, p_{1|1}\, p_{1|11}}$};
\draw[thick] (r) -- node[pos=.45, above, sloped, bitlab] {$0$} (n0);
\draw[thick] (r) -- node[pos=.45, above, sloped, bitlab] {$1$} (n1);
\draw[thick] (n0) -- node[pos=.45, above, sloped, bitlab] {$0$} (n00);
\draw[thick] (n0) -- node[pos=.45, above, sloped, bitlab] {$1$} (n01);
\draw[thick] (n1) -- node[pos=.45, above, sloped, bitlab] {$0$} (n10);
\draw[thick] (n1) -- node[pos=.45, above, sloped, bitlab] {$1$} (n11);
\draw[thick] (n00) -- node[pos=.45, above, sloped, bitlab] {$0$} (l000);
\draw[thick] (n00) -- node[pos=.45, above, sloped, bitlab] {$1$} (l001);
\draw[thick] (n01) -- node[pos=.45, above, sloped, bitlab] {$0$} (l010);
\draw[thick] (n01) -- node[pos=.45, above, sloped, bitlab] {$1$} (l011);
\draw[thick] (n10) -- node[pos=.45, above, sloped, bitlab] {$0$} (l100);
\draw[thick] (n10) -- node[pos=.45, above, sloped, bitlab] {$1$} (l101);
\draw[thick] (n11) -- node[pos=.45, above, sloped, bitlab] {$0$} (l110);
\draw[thick] (n11) -- node[pos=.45, above, sloped, bitlab] {$1$} (l111);
\end{tikzpicture}
\caption{A 3-qubit prefix tree. Each non-root node stores a conditional probability $p_{x_i|x_{<i}}$, and each leaf stores $\gamma_x$ times the product of conditional probabilities along its root-to-leaf path.}
\label{fig:aaronson-prefix-tree2}
\end{figure}

%% file: figures/mip-colored.tex
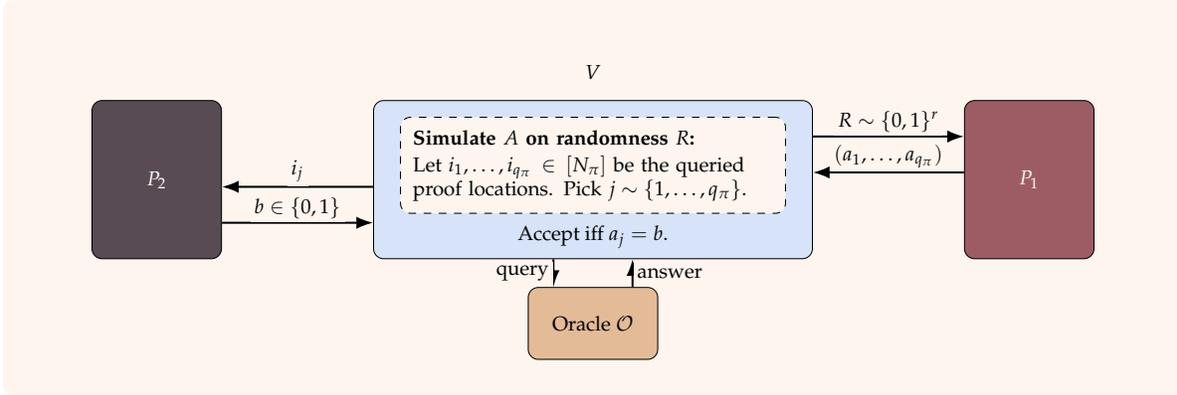
\begin{figure}[h]
\centering
\resizebox{0.95\linewidth}{!}{%
\begin{tikzpicture}[
    >=Latex,
    font=\footnotesize,
    party1/.style={
        draw, rounded corners,
        minimum width=1.8cm, minimum height=2.2cm,
        align=center, fill=ufcol
    },
    party2/.style={
        draw, rounded corners,
        minimum width=1.8cm, minimum height=2.2cm,
        align=center, fill=ugcol
    },
    oracle/.style={
        draw, rounded corners,
        minimum width=1.8cm, minimum height=1.0cm,
        align=center, fill=meascol
    },
    verifier/.style={
        draw, rounded corners,
        minimum width=6.1cm, minimum height=2.2cm,
        fill=hadcol
    },
    innerbox/.style={
        draw, dashed, rounded corners,
        align=left, text width=5cm,
        inner sep=5pt,
        fill=bgcol
    },
    msg/.style={-Latex, thick},
    lab/.style={fill=bgcol, inner sep=1.5pt, align=center}
]

\definecolor{bgcol}{HTML}{FEF5EF}
\definecolor{hadcol}{HTML}{D6E3F8}
\definecolor{ufcol}{HTML}{9D5C63}
\definecolor{ugcol}{HTML}{584B53}
\definecolor{meascol}{HTML}{E4BB97}

\fill[bgcol, rounded corners=4pt] (-8.2,-3.0) rectangle (8.2,2.5);

\node[verifier] (V) at (0,0) {};
\node[party1, right=2.1cm of V] (P1) {\textcolor{white}{$P_1$}};
\node[party2, left=2.1cm of V] (P2) {\textcolor{white}{$P_2$}};
\node[oracle] at (0,-2.0) (O) {Oracle $\cO$};

\node at ([yshift=1.5cm]V.center) {$V$};

\node[innerbox] at ([yshift=0.2cm]V.center) {%
\textbf{Simulate $A$ on randomness $R$:}\\[1pt]
Let $i_1,\ldots,i_{q_\pi}\in[N_\pi]$ be the queried proof locations. 
Pick $j\sim\{1,\ldots,q_{\pi}\}$.};

\node[align=center] at ([yshift=-0.8cm]V.center) {Accept iff $a_j=b$.};

\draw[msg] ([yshift=0.6cm]V.east) -- ([yshift=0.6cm]P1.west)
    node[midway, above, lab] {$R\sim\{0,1\}^{r}$};

\draw[msg] ([yshift=0.1cm]P1.west) -- ([yshift=0.1cm]V.east)
    node[midway, above, lab] {$(a_1,\ldots,a_{q_\pi})$};

\draw[msg] ([yshift=-0.1cm]V.west) -- ([yshift=-0.1cm]P2.east)
    node[midway, above, lab] {$i_j$};

\draw[msg] ([yshift=-0.6cm]P2.east) -- ([yshift=-0.6cm]V.west)
    node[midway, above, lab] {$b\in\{0,1\}$};

\draw[msg] ([xshift=-0.55cm]V.south) -- ([xshift=-0.55cm]O.north)
    node[midway, left, lab] {query};

\draw[msg] ([xshift=0.55cm]O.north) -- ([xshift=0.55cm]V.south)
    node[midway, right, lab] {answer};

\end{tikzpicture}%
}
\caption{The two-prover protocol induced by the adaptive PCP verifier $A$.}
\label{fig:mip-colored}
\end{figure}